 \documentclass[journal]{IEEEtran}
\usepackage[final]{graphicx}
\usepackage[reqno]{amsmath}
\usepackage{amssymb}
\usepackage{cite}
\usepackage{enumerate}
\usepackage{color}

\usepackage[lofdepth,lotdepth]{subfig}
\usepackage{graphicx}

\usepackage{booktabs}
\usepackage{multirow,multicol}
\usepackage{rotating}
\usepackage{bigstrut}

\usepackage{url}

\usepackage{algorithmic}
\usepackage{algorithm}

\usepackage{comment}

\includecomment{techreport}
\excludecomment{journal}

\includecomment{multicolumn}
\excludecomment{singlecolumn}

\def\cov{{\mathbb Cov}}
\def\expect{{\mathbb  E}}
\def\var{{\mathbb Var}}
\def\Pr{{\mathbb P}}

\def\eqdef{\triangleq}

\def\12{\frac{1}{2}}

\def\dist{\raise.17ex\hbox{$\scriptstyle\sim$}}

\newfont{\bbb}{msbm10 scaled 500}

\newfont{\bb}{msbm10 scaled 1100}

\newcommand{\dv}{{\bf d}}

\newcommand{\kv}{{\bf k}}
\newcommand{\lv}{{\bf l}}

\newcommand{\ov}{{\bf o}}

\newcommand{\sv}{{\bf s}}

\newcommand{\Lc}{{\cal L}}

\newcommand{\Nc}{{\cal N}}

\newcommand{\etav}{\hbox{\boldmath$\eta$}}

\newcommand{\phiv}{\hbox{\boldmath$\phi$}}

\newtheorem{theorem}{Theorem}
\newtheorem{lemma}{Lemma}
\newtheorem{definition}{Definition}

\allowdisplaybreaks

\author{
Onur Gungor, Fangzhou Chen, C. Emre Koksal \\
\begin{tabular} {c}
\small Department of Electrical and Computer Engineering\\
\small The Ohio State University, Columbus, 43210\\
\end{tabular}
\thanks{
The authors are with the Department of Electrical and Computer Engineering,
The Ohio State University, Columbus, OH, 43210.
This work was
in part presented in the Workshop on Physical Layer Security \cite{conference}, Globecom 2011.

This work is supported in part by QNRF under grant NPRP 5-559-2-227,
and by NSF under grants CNS-1054738, CNS-0831919, CCF-0916664.
}
}

\title{Secret Key Generation Via Localization and Mobility}

\begin{document}
\maketitle


\begin{abstract}
We consider secret key generation from relative
localization information of a pair of nodes in a mobile wireless
network in the presence of a mobile eavesdropper. Our
problem can be categorized under the source models of information
theoretic secrecy, where the distance between the legitimate
nodes acts as the observed common randomness. We characterize
the theoretical limits on the achievable secret key bit rate, in terms of the observation
noise variance at the legitimate nodes and the eavesdropper. This
work provides a framework that combines information theoretic
secrecy and wireless localization, and proves that the localization
information provides a significant additional resource for secret
key generation in mobile wireless networks.
\end{abstract}


\section{Introduction}

We consider the generation of a common key in a pair of nodes, which move in $\mathbb{R}^2$ (continuous space)
 according to a stochastic mobility model. We exploit the reciprocity of the distance between a 
given pair of locations, view the distance between the legitimate nodes as a common randomness shared by these 
nodes and utilize it to generate secret key bits using the ideas from source models of secrecy \cite{Maurer}.

Unlike the recent plethora of studies (see Section~\ref{sec:relatedwork} for a brief list of related papers) that focuses on
wireless channel reciprocity, a variety of technologies can be used for localization  (e.g., ultrasound, infrared, Lidar, Radar, wireless radios),
which makes distance reciprocity an \emph{additional resource} for generating secret key bits.
Such versatility makes the key generation systems more robust, since different technologies may have different capabilities that wireless RF does not have. 
For instance, narrow beam width of infrared systems would make them 
less susceptible to eavesdropping from different angles.
Distance reciprocity is highly robust, since the distance measured between any pair of points is identical, regardless of which point the measurement originates. (e.g., when there is no line-of sight, or
when different frequency bands are used each way).
Yet, there are various challenges in obtaining reciprocal distance measurements. 

In this paper, we analyze the theoretical limits of key generation using localization in the following system. We assume that mobile nodes
 obtain observations regarding the sequence of distances between them over a period of 
time as they move in the area. 
The measurements can be obtained actively through exchange of wireless radio, ultrasound, infrared beacons, or passively by processing existing video images, etc.
The beacon signal may contain explicit information such as a time stamp, or the receiving node can extract other
 means of localization information by analyzing angle of arrival, received signal strength, etc. 
The nodes perform localization based on the observations of distances, and the statistics of the mobility model, and obtain
estimates of their relative locations with respect to each other. 
Then, the nodes
communicate over the public channel to agree on a secret key. 
The generated key bits satisfy the following three quality measures: i) reliability, ii) secrecy, and iii) randomness.
For reliability, we show that the probability of mismatch between the keys generated by the legitimate nodes decays to $0$ with increasing block length.
In our attacker model, we consider a passive eavesdropper, that 
overhears the exchanged beacons in the first phase, and the public discussion in the third phase, and 
tries to deduce the generated key based
solely on these observations. 
The attacker can follow various mobility strategies in order to enhance its position statistically to reduce the
 achievable key rate (possibly to $0$). We assume that the attacker does not actively interfere with the observation phase,
 e.g., by injecting jamming signals, etc., in order not to reveal its presence.
For secrecy, we consider Wyner's notion, i.e., the rate at which mutual information on the key leaks to the eavesdropper should be arbitrarily low.
For randomness, the generated key bits have to be perfectly compressed, i.e., the entropy should be equal to the number of bits it contains.

We mainly focus on information theoretic limits. 
Using a source model of secrecy \cite{Maurer}, we characterize the achievable secret key bit rate in terms of observation noise parameters 
at the legitimate nodes and the eavesdropper under two different cases of global location information (GLI): (i) No GLI, 
in which the nodes do not observe their global locations directly, and (ii) perfect GLI, in which nodes have perfect observation 
of their global locations, through a GPS device, for example. 
While the bounds we provide are general for a large set of observation statistics,  we 
 further investigate the scenario in which
the observation noise is i.i.d. Gaussian for all nodes: We study the observation SNR asymptotics, and
 show a phase-transition phenomenon for the key rate. In particular, we prove that 
the secret key rate grows unboundedly as the observation noise variance decays, if the eavesdropper does not obtain the angle of arrival observations. 
Otherwise, it is not possible to increase the secret key rate beyond a certain limit.
Then, we evaluate the theoretical performance numerically for a simple grid-type model, 
as a function of beacon power.
We also evaluate the performance for the case where the eavesdropper 
strategically changes its location to reduce the secret key rate.
Specifically, we consider the strategy where the eavesdropper
moves to the middle of its location estimates of the legitimate nodes.
We show that with this strategy, the eavesdropper can significantly reduce the secret key rate
compared to the case where it follows a random mobility pattern.

In summary, our main contribution is
to illustrate that relative localization information can be used as an additional resource 
for secret key generation (see Section~\ref{sec:relatedwork} for a comparison with related work).

\color{black}
\subsection{Related Work}\label{sec:relatedwork}
\label{sec:related}
\vspace{-0.05in}

Generation of secret key from relative localization information can be categorized under 
source model of information theoretic secrecy, which studies generation of secret key bits 
from common randomness observed by legitimate nodes.
In his seminal paper \cite{Maurer}, Maurer showed that, if two nodes observe 
correlated randomness, then they can
agree on a secret key through public discussion.
He provided upper and lower bounds on the achievable secret key rates. 
Although the bounds have been improved later \cite{Ahlswede:93,Maurer:99}, the secret key capacity of the source model in general is
still an open problem. Despite this fact, the source model has been utilized in several different settings \cite{Csiszar:08,Maurer:03,Csiszar}.

There is a vast amount of literature on localization (see, e.g., \cite{Gezici,Shen1} for 
wireless localization, \cite{infrared} for infrared localization,
 and \cite{ultrasound} for  ultrasound localization).
There has been some focus on secure localization and position-based cryptography~\cite{buhrman,srinivasan,poovendran,chandran}, 
however, these works either consider key generation in terms of other forms of secrecy (i.e., computational secrecy), 
or fall short of covering a complete information theoretic analysis.

A similar line of work in wireless network secrecy considers channel
identification~\cite{Tse} for secret key generation using wireless radios.
Based on the channel reciprocity assumption, nodes at both ends experience the same channel,
corrupted by independent noise. Therefore, nodes can use their channel magnitude and phase response observations to generate secret key bits from
public discussion. 
The literature on channel identification based secret key generation is vast.
The works \cite{channel_id1,channel_id2,channel_id3,channel_id4,channel_id5,channel_id6} study key generation with on-the-shelf devices, under 802.11 development platform
using a two way radio signal exchange on the same frequency.
\cite{proximate}, on the other hand, utilizes the fact that fading is highly correlated on locations that are less than a half wavelength apart,
instead of exploiting the reciprocity. Therefore,
very close nodes can use public radio signals (e.g., FM, TV, WiFi) to generate secret key bits.

In most of these works, the security analysis is based on the assumption that the channel gains are modeled as random processes, 
that are independent of the distances between the nodes, and are independent at locations that are more than a few wavelengths apart.
While being appropriate for a non line-of-sight and highly dynamic media, these models do not capture wireless propagation in environments where attenuation is a function of the propagation distance.
 In such environments, an attacker that has some localization capabilities will gain a statistical advantage by estimating the channel gains based on its distance observations. If the key generation process ignores this advantage, part of the key may be recovered by the attacker and thus the key cannot be perfectly secure. For instance, Jana et. al. \cite{channel_id2} focuses on a scenario in which
secret key bits based on the received signal strength (RSSI), and show that an eavesdropper that knows the
location of the legitimate nodes can launch a mobility attack to force the legitimate nodes to generate deterministic key bits, by periodically blocking and un-blocking
their line-of-sight. Similarly, if the eavesdropper is close (less than a wavelength) to one of the legitimate nodes, then eavesdropper will obtain correlated
information \cite{proximate},
therefore the generated key will not be \emph{perfectly secure}, and secrecy outage occurs.
The practical applicability of exploiting channel reciprocity for secret key generation has also been questioned 
recently in \cite{channel_id7}. It is shown that, especially when the nodes have
sufficient mobility, the eavesdropper's and the legitimate receiver's channel can be significantly correlated depending on the locations, which breaks the secrecy of the initial generated key.

On the other hand,
key generation based on locations does not make such independence assumptions. The dependencies in the locations of the legitimate nodes and the observations of the attacker with those of the legitimate nodes are taken into account to provide \emph{provably security} against a mobile eavesdropper
with localization capability.
Thus, the insights provided in this paper can also be valuable for the class of studies on key generation based on wireless channel reciprocity, as we show how one should capture a variety of capabilities of the attackers in finding the correct rate for the key and in designing the appropriate mechanisms to generate a truly secret key. 

A word about notation: We use $[x]^+ = \max(0,x)$ and  $\|.\|$ denotes the L2-norm.
A brief list of variables used in the paper can be found in Table~\ref{tab:variables}.


\section{System Model}
\label{sec:Systemmodel}

\subsection{Mobility Model}
We consider a simple network consisting of two mobile legitimate nodes, called user $1$ and $2$, and
a possibly mobile eavesdropper $e$. We divide time uniformly into $n$ discrete slots.
Let $l_j[i] \in \Lc$ be the random variable that denotes the coordinates of the location of node $j \in \{1,2,e\}$ in slot $i\in \{1,\cdots, n\}$,
where nodes are restricted to the field $\Lc \subset \mathbb{R}^2$.
We use the boldface notation
$\lv_j = \{ l_j[i]\}_{i=1}^n$, to denote the $n$-tuple location vectors for $j \in \{ 1, 2, e \}$.
The distance between nodes $1$ and $2$ in slot $i$ is ${d}_{12}[i] = \|l_{1}[i]-l_2[i]\|$.
Similarly, ${d}_{1e}[i]$ and ${d}_{2e}[i]$ denote the sequence of distances between nodes $(1,e)$ and nodes $(2,e)$ respectively.
We use the boldface notation $\dv_{12}$, $\dv_{1e}$, $\dv_{2e}$ for the $n$-tuple distance vectors.
Note that, in any slot the nodes form a triangle in $\mathbb{R}^2$, as depicted in Figure~\ref{fig:sysmodel},
where $\phi_{12}[i]$ , $\phi_{21}[i]$, $\phi_{1e}[i]$, $\phi_{2e}[i]$
denote the angles with respect to some coordinate axis.
\begin{figure}[ht]
    \centering
     \includegraphics[scale =0.7]{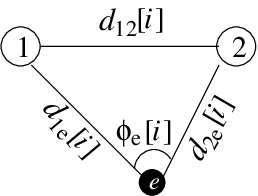}
    \caption{Legitimate nodes and the eavesdropper form a triangle.}
     \label{fig:sysmodel}
\end{figure}
We assume that the distances $d_{12}[i], \ d_{1e}[i], \ d_{2e}[i]$ take values in the interval $[d_{\min}~d_{\max}]$,
since the nodes cannot be closer to each other than $d_{\min}$ due to physical restrictions, and they cannot be further than $d_{\max}$
away from each other due to their limited communication range.
We assume that the location vectors $\lv_1,\lv_2,\lv_e$ are ergodic processes.
We will use the notation $\sv \eqdef [\lv_1,\lv_2,\lv_e]$ to summarize the state variables
related to mobility in the system. Note that $s[i] \in \Lc^3 = \Lc \times \Lc \times \Lc$ for any $i$
\footnote{It is not necessary to use absolute coordinates for $\lv_1,\lv_2,\lv_e$. For example, 
when global locations are not available at the nodes, we may assume that node $1$ is at the origin, i.e., $l_1[i] = [0~0]$ for all $i$}.

\begin{table}[htbp]
  \centering
  \caption{List of variables}
    \begin{tabular}{lll}
    \toprule
    var. & Description  \\
    \midrule
    $n$      &  number of slots \\
    $T$      &  number of steps in public discussion \\
    $d_{ij}$ &  distance between nodes $i$ and $j$ \\
    $\Lc \in \mathbb{R}^2$    & the field where nodes are located \\
    $l_j$    &  2-D location of node $j$  \\
    $\phi_{ij}$ & angle between nodes $i$ and $j$ \\
    $\hat{d}_j, \hat{\phi}_j$ & observation of nodes $j \in \{1,2\}$ of $d_{12}$ and $\phi_j$ \\
    $\hat{d}_{je}$ & observation of node $e$ of $d_{je}$ and $\phi_{je}$ \\
    $o_j$ & complete observations of node $j$ based on available GLI \\
	$s$ & location triple $[l_1,l_2,l_e]$ \\
    $s^{\Delta}$ & quantized version of location triple, $[l_1^{\Delta},l_2^{\Delta},l_e^{\Delta}]$ \\
    $\Delta$ & quantization resolution \\
    $\psi$   & uniform $2-D$ quantization function \\
    $\tilde{s}_j^{\Delta}$ & $[\tilde{l}_{1,j}^{\Delta},\tilde{l}_{2,j}^{\Delta},\tilde{l}_{e,j}^{\Delta}]$ \\
    $\tilde{l}_{k,j}^{\Delta}$ & node $j$'s estimate of $l_k^{\Delta}$ based on all its information \\
    $\kappa(.,m)$ & $m$-bit Gray coder \\
    $v_j$  & obtained binary key at node $j$ before reconciliation \\
	$u_j$ & obtained binary key at node $j$ after reconciliation  \\
	$q_j$ & obtained binary key at node $j$ after universal compression  \\
	$k_j$  & final key at node $j$ after universal hashing \\	
    \bottomrule
    \end{tabular}%
  \label{tab:variables}%
\end{table}%

\subsection{Localization}\label{s:localization}

At each time slot, there is a period in which the legitimate nodes 
obtain information about their relative position with respect to each other.
As discussed in Section~\ref{sec:relatedwork}, there are various methods to establish
the localization information. 
In this paper, we will not treat these methods separately. We will simply assume that,
during measurement period $i$, when node $1$ transmits a beacon, 
nodes $2$ and $e$ obtain a noisy observation of $d_{12}[i]$ and $d_{1e}[i]$
respectively. Let these observation be $\hat{d}_{2}[i]$ and $\hat{d}_{1e}[i]$, respectively. 
Similarly, when node $2$ follows up with a beacon, nodes $1$ and $e$ obtain the distance observations $\hat{d}_{1}[i]$ and $\hat{d}_{2e}[i]$,
respectively. 
The nodes may also independently observe their global positions, e.g., through a GPS device. 
They may also observe the angle they make with respect to each other, if they are equipped with direction sensitive localizers (e.g., directional antennas in wireless localization).
We consider two extreme cases on the global location information (GLI):\\
1) no GLI: The nodes do not have any knowledge of their global location. 
However,  with the observations of both the beacons, the eavesdropper also obtains a noisy
observation, $\hat{\phi}_e[i]$, of the angle between the legitimate nodes.\\
2) perfect GLI: Each node has perfect knowledge of its global location, and
a sense of orientation with respect to some
coordinate plane as shown in Figure~\ref{fig:sysmodel2}. In this case,
nodes $1$, $2$ obtain noisy observations $\hat{\phiv}_1, \ \hat{\phiv}_2$
of the angle $\phiv_{12}$. Similarly, node $e$ obtains noisy observation
$\hat{\phiv}_{1e}, \ \hat{\phiv}_{2e}$ of the angles $\phiv_{1e}, \phiv_{2e}$.

Let $o_j[i]$ denote the set of observations of node during slot $i$, and
 $\ov_j \eqdef \{o_j[i]\}_{i=1}^n$. The observations $\ov_j$ for each case is provided in Table~\ref{tab:observations}.
We emphasize that, the observations in each slot are obtained  solely from the beacons exchanged 
during that particular slot. The nodes' final estimates of the distances depend also
on the observations during other slots, due to predictable mobility patterns. 
\begin{figure}[ht]
    \centering
     \includegraphics[scale = 0.7]{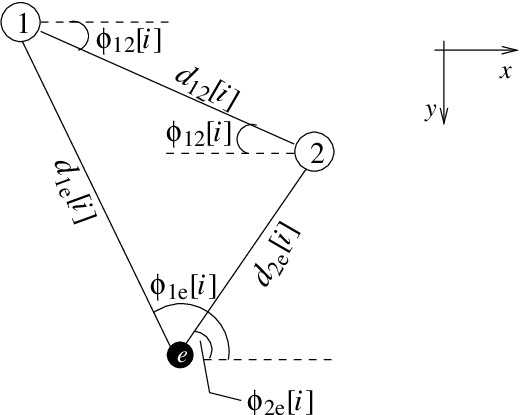}
    \caption{With GLI, the nodes obtain noisy observations of the relative orientation of each other with respect to the x-axis.}
     \label{fig:sysmodel2}
\end{figure}
\begin{table}[htbp]
  \centering
  \caption{Nodes' Observations}
    \begin{tabular}{l|ll}
    \toprule
    & No GLI & Perfect GLI  \\ 
    \midrule
    $o_1[i]$ & $[\hat{d}_1[i]]$ & $[\hat{d}_1[i],\hat{\phi}_1[i],l_1[i]]$  \\
    $o_2[i]$ & $[\hat{d}_2[i]]$ & $[\hat{d}_2[i],\hat{\phi}_2[i],l_2[i]]$ \\
    $o_e[i]$ & $[\hat{d}_{1e}[i],\hat{d}_{2e}[i],\hat{\phi}_e[i]]$ & $[\hat{d}_{1e}[i],\hat{d}_{2e}[i],\hat{\phi}_{1e}[i], \hat{\phi}_{2e}[i], l_e[i]]$  \\
    \bottomrule
    \end{tabular}%
  \label{tab:observations}%
\end{table}%
\subsection{Attacker Model}\label{sec:attacker}

We assume that there exists a passive eavesdropper $e$, which does not transmit any beacons.
However, node $e$ can strategically change its location to
obtain a geographical advantage against the legitimate nodes. Overall, we consider two strategies:\\
\textbf{Random Mobility:} Eavesdropper moves randomly, without a regard 
to the location of the legitimate nodes. 
We will assume
that eavesdropper adopts random mobility unless otherwise stated. \\
\textbf{Mobile Man in the Middle:}
Node $e$ \emph{controls} its mobility, such that it can move accordingly 
to obtain a geographic advantage
compared to legitimate nodes. 
We consider the strategy where node $e$ moves to the mid-point of
its maximum likelihood estimates of the legitimate nodes' locations. 
For $j\in\{1,2\}$, let us denote node $e$'s maximum likelihood estimate of node $j$'s location
at slot $i$, based on its observations \emph{up to slot} $i-1$ as $\tilde{l}_{j,e}[i]$. Then,
$$\tilde{l}_{j,e}[i] = \arg\max_{l_j[i] \in \Lc} \Pr(l_j[i] | o_e[1],\ldots,o_e[i-1])$$
In other words, node $1$ and node $2$'s locations at slot $i$ is predicated by node $e$
by its observations in the previous slots. 
Then, at the beginning of each slot $i$, node $e$ moves to the mid-point of
the estimates, which is
$(\tilde{l}_{1,e}[i] + \tilde{l}_{2,e}[i])/2$.

Although we restricted ourselves to a single passive eavesdropper,
we also discuss the implications of multiple eavesdroppers.
The eavesdroppers may utilize their observations in two possible ways:
(i) \emph{Non-colluding eavesdroppers} do not communicate, or share their observations with each other,
whereas (ii) \emph{colluding eavesdroppers} combine their measurements to obtain less noisy measurements.
Note that, an eavesdropper with multiple location sensors 
(e.g., multiple antennas in the case of wireless radio-based localization) is a special case of colluding eavesdroppers, as
each sensor could be viewed as a separate eavesdropper, with perfect links between them.
Theoretical secret key capacity under colluding eavesdropper scenario is lower, due to cooperation of the eavesdroppers,
as discussed in Section~\ref{sec:keyrates}.



\subsection{Notion of security}

We consider the typical definition of source model of 
information theoretic secrecy under a passive eavesdropper:
We assume that there exists an authenticated error-free 
public channel, using which the legitimate nodes can communicate to agree on secret keys, based on the
observations of the distances and angles ($\ov_1$ and $\ov_2$) obtained during beacon exchange.
This process, commonly referred to as public discussion~\cite{Maurer},
is a $T$ step message exchange protocol,
where at any step $t \in \{1,\cdots,T\}$, node $1$ sends message $C_1[t]$, and node $2$ replies back with message $C_2[t]$
such that, for $t>1$,
\begin{align}
H\big{(}C_1[t]|{\ov}_{1},\{C_{1}[i]\}_{i=1}^{t-1},\{C_{2}[i]\}_{i=1}^{t-1}\big{)} &= 0,\mbox{ odd }t \label{eq:ci1}\\
H\big{(}C_2[t]|{\ov}_{2},\{C_{1}[i]\}_{i=1}^{t},\{C_{2}[i]\}_{i=1}^{t-1}\big{)} &= 0,\mbox{ even }t. \label{eq:ci2}
\end{align}
 At the end of the $T$ step protocol, node $1$ obtains $\kv_1$,
and node $2$ obtains $\kv_2$ as the secret key, where
\begin{align}
H\left(\kv_j|{\ov}_j,\{C_1[t],C_2[t]\}_{t=1}^T\right) &= 0,~ j \in \{1,2\}. \label{eq:si}
\end{align}
\begin{definition}\label{def:security}
Secret key bits are generated (with respect to the described attacker model) at rate $R$, if, for all $\epsilon >0$ and $\delta>0$, 
there exists some $n,\ T >0$ such that \eqref{eq:ci1}, \eqref{eq:ci2} and \eqref{eq:si} are satisfied, and
\begin{align}
H(\kv_j)/n &= R,~ j \in \{1,2\}  \label{randomness} \\
\Pr(\kv_1 \neq \kv_2) &\leq \epsilon \label{reliability}\\
I(\kv_j;{\ov}_e,\{C_1[t],C_2[t]\}_{t=1}^T)/n &\leq \delta,~ j \in \{1,2\}. \label{equivocation}
\end{align}
\end{definition}
Here, \eqref{randomness}-\eqref{equivocation} correspond to perfect randomness, reliability and 
security constraints, respectively.
The schemes proposed in the literature typically use a random coding structure,
where $\{C_1[t],C_2[t]\}_{t=1}^T$ are generated by using a binning strategy~\cite{Maurer}-\cite{Csiszar}.
In Section~\ref{sec:keyrates}, we will make use of these existing results to provide computable theoretical
bounds on the achievable key rates. 


\section{Theoretical Performance Limits}\label{sec:keyrates}
In this section, we provide information theoretical bounds on the achievable key rate with perfect reliability. To evaluate these bounds, we assume an idealized system by ignoring the issues associated with quantization, cascade reconciliation protocol, and privacy amplification. Thus, these bounds are valid for \emph{any} key generation scheme that satisfies Definition~\ref{def:security}.
\begin{theorem} \label{t:theoreticalbounds}
A lower bound $R_L$, and an upper bound $R_U$ on the perfectly-reliable key rate achievable
through public discussion are
\begin{multicolumn}
\begin{align}
&R_L =
\max\bigg{\{}\lim_{n \to \infty}\frac{1}{n}\left[I({\ov}_{1};{\ov}_{2})-  I({\ov}_{1}; \ov_e)\right]^+,  \nonumber\\
 &\qquad \lim_{n \to \infty}\frac{1}{n} \left[I(\ov_{2};\ov_{1})- I(\ov_{2};{\ov}_e) \right]^+ \bigg{\}} \label{RL}
\end{align}
\begin{align}
R_U =
\lim_{n \to \infty}\frac{1}{n}\min\left\{I(\ov_{1};\ov_{2}),
I(\ov_{1};\ov_2|\ov_e)\right \}  \label{RU}
\end{align}
\end{multicolumn}
\begin{singlecolumn}
\begin{align}
&R_L =
\max\bigg{\{}\lim_{n \to \infty}\frac{1}{n}\left[I({\ov}_{1};{\ov}_{2})-  I({\ov}_{1}; \ov_e)\right]^+,  
  \lim_{n \to \infty}\frac{1}{n} \left[I(\ov_{2};\ov_{1})- I(\ov_{2};{\ov}_e) \right]^+ \bigg{\}} \label{RL} \\
  &R_U =
\lim_{n \to \infty}\frac{1}{n}\min\left\{I(\ov_{1};\ov_{2}),
I(\ov_{1};\ov_2|\ov_e)\right \}  \label{RU}
\end{align}
\end{singlecolumn}
respectively, where $\ov_1, \ov_2$ and $\ov_e$
are as given in Table~\ref{tab:observations} for different possibilities of GLI.
\end{theorem}
The theorem follows\footnote{Theorem 4 of \cite{Bloch}
provide general upper and lower bounds including the case where the source
processes are not ergodic. In our system model, $\ov_1, \ \ov_2$ and $\ov_e$ are
ergodic processes, hence they are information stable, therefore
these lower and upper bounds reduce to \eqref{RL} and \eqref{RU}, respectively \cite{Verdu}.} from Theorem 4 in~\cite{Bloch}, which generalizes
Maurer's results on secret key generation through public discussion \cite{Maurer},
to non-i.i.d. settings. Although tighter bounds exist in the literature \cite{Ahlswede:93,Maurer:99},
we use the above bounds since they provide clearer insights into our systems due to their simplicity. 

Note that for the special case where 
the observations $(o_1[i],o_2[i],o_e[i])$
are i.i.d., we can safely drop the index $i$, and denote the joint probability density function
of observations as $f(o_1,o_2,o_e)$. Therefore, the conditioning on the past and future observations
in $R_L$ and $R_U$ disappear, and the bounds reduce to
\begin{multicolumn}
\begin{align}
R_L = &\max\bigg{(}\left[I(o_1;o_2)-I(o_1;o_e)\right]^+,  \nonumber\\
      & \left[ I(o_1;o_2)-I(o_2;o_e)\right]^+ \bigg{)} \label{RLiid} \\
R_U = & \min\left(I(o_1;o_2),I(o_1;o_2|o_e)\right),  \label{RUiid}
\end{align}
\end{multicolumn}
\begin{singlecolumn}
\begin{align}
R_L = &\max\bigg{(}\left[I(o_1;o_2)-I(o_1;o_e)\right]^+, \left[ I(o_1;o_2)-I(o_2;o_e)\right]^+ \bigg{)} \label{RLiid} \\
R_U = & \min\left(I(o_1;o_2),I(o_1;o_2|o_e)\right),  \label{RUiid}
\end{align}
\end{singlecolumn}

Also note that Theorem~\ref{t:theoreticalbounds} can be extended 
to provide key rate bounds against multiple eavesdropper models discussed in Section~\ref{sec:attacker}. 
Consider $K$ eavesdroppers, with observations $\ov_{e,1},\ldots, \ov_{e,K}$. 
For the non-colluding eavesdroppers model, 
since the eavesdroppers are not communicating, we can safely consider the most capable eavesdropper $k$. 
In other words, in \eqref{RL}, \eqref{RU} we can replace $\ov_e$ with $\ov_{e,k}$ for $k\in\{1,\ldots K\}$ which yields the lowest bounds, 
and discard the rest of the eavesdroppers.
For the colluding eavesdroppers model, we can replace the term $\ov_e$ in \eqref{RL}, \eqref{RU} with
 $\ov_{e,1},\ldots, \ov_{e,K}$ since the eavesdroppers perfectly communicate with each other.
It can be directly observed that the bounds for colluding case are lower with respect to the non-colluding case.

\section{Gaussian observations}\label{sec:Gaussian}

To obtain more insights from theoretical results in Section~\ref{sec:keyrates}, we focus on the 
following special case: First, we assume that the node locations are individually Markov processes such that
$$l_j[i-1]\rightarrow l_j[i] \rightarrow l_j[i+1],~ j \in \{1,2,e\},$$
holds for any $i$, and their joint probability density function
$f(\lv_1,\lv_2,\lv_e)$ is well defined.
Secondly, all observations of distance and angle terms are i.i.d. Gaussian processes.
This model is typically used in the literature to characterize
 observation noise \cite{Jourdan, Gezici}. 
To that end, for no GLI, $j\in\{1,2\}$
\begin{align}
\hat{d}_j[i]     & = d_{12}[i] +  w_j[i],
~ w_j[i] \dist \Nc\left( 0, \frac{\gamma(d_{12}[i])\rho_{j}}{P} \right) \label{eq:legit_obs}\\
\hat{d}_{je}[i]  & = d_{je}[i] +  w_{je}[i], 
~ w_{je}[i]  \dist \Nc\left( 0, \frac{\gamma(d_{je}[i])\rho_{e}}{P} \right)  \label{eq:eve_obs} \\
\hat{\phi}_{e}[i]& = \phi_{e}[i] + w_{\phi_e}[i],
~ w_{\phi}[i]   \dist \Nc\left( 0,  \frac{\gamma_{\phi}(d_{1e}[i],d_{2e}[i])\rho_{e}}{P}
\right)  \label{eq:eve_angle_obs}  
\end{align}
are Gaussian noise processes,
where $P$ is the beacon power. The observation noise variances are increasing functions
of the distance, which are modeled by the increasing functions $\gamma$ for
distance observations and $\gamma_{\phi}$ for angle observations.
The parameter $\rho_j$ depends on the capability of the nodes.
For instance, in wireless localization, $\gamma$ and $\gamma_{\phi}$
depend on the path loss exponent, and $\rho$ depends on
receiver antenna gain, number of antennas, etc \cite{Jourdan, Gezici}.
For perfect GLI, we additionally assume\footnote{For perfect GLI, the angle information is obtained according to a fixed coordinate plane, hence the function $\gamma_{\phi}$ has single argument.} that for $j \in \{1,2\}$,
\begin{align}
\hat{\phi}_{j}[i]    & = {\phi}_{j}[i]  +  w_{\phi_j}[i], 
~ w_{\phi_j}[i]	   \dist \Nc\left(0, \frac{\gamma_{\phi}(d_{12}[i])\rho_{j}}{P} \right) \label{eq:gli_legit_angle_obs}\\
\hat{\phi}_{je}[i]  & = {\phi}_{je}[i] +  w_{{\phi}_{je}}[i],  
~ w_{{\phi}_{je}}[i] \dist \Nc\left(0, \frac{\gamma_{\phi}(d_{je}[i])\rho_{e}}{P}\right) \label{eq:gli_eve_angle_obs} 
\end{align}
Clearly, the achievable key rates depend highly on the functions $\gamma$, $\gamma_\phi$ and $\rho$. 
Note that,
there there may be a bias on these observations due to small scale fading \cite{Jourdan}. The effect of biased observations are 
\begin{journal}
considered in our technical report \cite{techreport}.
\end{journal}
\begin{techreport}
studied in Appendix~\ref{app:practical}.
\end{techreport}

\subsection{Beacon Power Asymptotics}\label{sec:beacon_power}

In this part, we analyze the beacon power asymptotics of the system.
We show that, if the eavesdropper does not observe the angle\footnote{Note that in some cases, the nodes cannot obtain any useful angle information, e.g.,
in wireless localization, when each node is equipped with a single omnidirectional antenna.}, i.e.,
$\hat{\phi}_e = \emptyset$, then $R_L$ increases unboundedly with the
beacon power $P$, which indicates that arbitrarily large secret key rates can be obtained. However,
when eavesdropper observes the angle information, then $R_U$ remains bounded, which indicates that the advantage gained
by increasing beacon power is rather limited. 
To clearly illustrate our insights, we present our results for the no GLI scenario. However,
the same conclusion holds for the perfect GLI case as well.
\begin{theorem}\label{t:RUbound}
When the eavesdropper obtains angle information, i.e., $I(\hat{\phiv}_e;\phiv_e)>0$,
\begin{align}
\lim_{P \to \infty} R_U < \infty.
 \end{align}
\end{theorem}
The proof is in 
\begin{journal}
Appendix~\ref{app:keyratebounds},
\end{journal}
\begin{techreport}
Appendix~\ref{app:keyratebounds}, 
\end{techreport}
where we show that $\lim_{P \to \infty} R_U \leq \eta$,
where
\begin{singlecolumn}
\begin{align}
\eta &= \frac{1}{2} \log\bigg{\{}
  2\pi\expect \bigg{[} \frac{\rho_e}{d_{12}^2} \bigg{(}
 \frac{d_{12}^2\rho_1}{\rho_e} \gamma(d_{12})  +  4(d_{1e}+d_{2e})^2 (\sqrt{\gamma(d_{1e})}
+\sqrt{\gamma(d_{2e})})^2 \nonumber\\
& ~+ (4 d_{1e}d_{2e}+ 64(d_{1e}d_{2e})^2)\gamma_{\phi}(d_{1e},d_{2e}) + 8(d_{1e}+d_{2e})d_{1e}d_{2e}\big{(}\sqrt{\gamma(d_{1e})} \nonumber\\
& ~ +\sqrt{\gamma(d_{2e})} \big{)}\sqrt{\gamma_{\phi}(d_{1e},d_{2e})}
   + 64 (d_{1e}+d_{2e})^2 (\gamma(d_{1e})+\gamma(d_{2e})) \bigg{)} \bigg{]} \bigg{\}} 
 - \frac{1}{2}\expect\bigg{[}\log\left(2\pi\rho_1\gamma( d_{12})\right) \bigg{]}. \label{RUbound_eta}
\end{align}
\end{singlecolumn}
\begin{multicolumn}
\begin{align}
\eta &= \frac{1}{2} \log\bigg{\{}
  2\pi\expect \bigg{[} \frac{\rho_e}{d_{12}^2} \bigg{(}
 \frac{d_{12}^2\rho_1}{\rho_e} \gamma(d_{12}) \nonumber\\
& +  4(d_{1e}+d_{2e})^2 (\sqrt{\gamma(d_{1e})}
+\sqrt{\gamma(d_{2e})})^2 \nonumber\\
& + (4 d_{1e}d_{2e}+ 64(d_{1e}d_{2e})^2)\gamma_{\phi}(d_{1e},d_{2e}) \nonumber\\
&  + 8(d_{1e}+d_{2e})d_{1e}d_{2e}\big{(}\sqrt{\gamma(d_{1e})} \nonumber\\
&  +\sqrt{\gamma(d_{2e})} \big{)}\sqrt{\gamma_{\phi}(d_{1e},d_{2e})}
 \nonumber\\
&+ 64 (d_{1e}+d_{2e})^2 (\gamma(d_{1e})+\gamma(d_{2e})) \bigg{)} \bigg{]} \bigg{\}} \nonumber\\
& - \frac{1}{2}\expect\bigg{[}\log\left(2\pi\rho_1\gamma( d_{12})\right) \bigg{]}. \label{RUbound_eta}
\end{align}
\end{multicolumn}
The parameter $\eta$ remains finite since the distances take on values in 
some bounded range $[d_{\min},d_{\max}]$ with probability $1$. Therefore, the secret key rate remains bounded.
\begin{theorem}\label{t:RLbound}
When the eavesdropper does not obtain any angle information, i.e.,
$I(\hat{\phiv}_e;\phiv_e)=0$, 
\begin{align*}
\lim_{P \to \infty} \frac{R_L}{\frac{1}{2} \log (P) } =  \lim_{P \to \infty} \frac{R_U}{\frac{1}{2} \log (P) } = 1.
\end{align*}
\end{theorem}
The proof is provided in 
\begin{techreport}
Appendix~\ref{app:keyratebounds}.
\end{techreport}
\begin{journal}
Appendix~\ref{app:keyratebounds}.
\end{journal}
Theorem~\ref{t:RLbound} implies that, without the angle observation at the eavesdropper, an arbitrarily large key rate can be achieved with
sufficiently large beacon power $P$.
However, the key rate increases with $\log(P)$, which means that increasing
the beacon power would provide diminishing returns.
\subsection{Numerical Evaluations}\label{s:simulations}
We evaluate 
the theoretical bounds in Section~\ref{sec:keyrates} for Gaussian observations model, using
Monte Carlo simulations. 
\subsubsection*{Setup}
We consider a simple $(M \times M)$ discrete 2-D grid, which
simulates a city with $M$ blocks that covers a square field of area $A^2$, such that
for any $j \in \{1,2,e\}$, $i \in\{1,\ldots,n\}$, 
$l_j[i] = [x~y] \in \{\frac{A}{M},\ldots,M\frac{A}{M}\}
\times \{\frac{A}{M},\ldots,M\frac{A}{M}\}$.  
Node mobilities are Markov, and characterized by parameter $B$, where
\begin{singlecolumn}
\begin{align*}
\Pr \big{(}l_j[i] = [x~y] ~|~ & l_j[i-1] = [x'~y']\big{)} =
\begin{cases}
\frac{1}{(B+1)^2} \mbox{ if} & |x-x'| \leq \frac{AB}{M} \\
                             & |y-y'| \leq \frac{AB}{M} \\
~ ~ 0               &\mbox{ otherwise }
\end{cases}
\end{align*}
\end{singlecolumn}
\begin{multicolumn}
\begin{align*}
\Pr \big{(}l_j[i] = [x~y] ~|~ & l_j[i-1] = [x'~y']\big{)} =
 \\ &
\begin{cases}
\frac{1}{(B+1)^2} \mbox{ if} & |x-x'| \leq \frac{AB}{M} \\
                             & |y-y'| \leq \frac{AB}{M} \\
~ ~ 0               &\mbox{ otherwise }
\end{cases}
\end{align*}
\end{multicolumn}
For no GLI, we choose
$\gamma(d) = 0.1 + d^2$, and 
$$\gamma_{\phi}(d_{1e},d_{2e}) = \pi - \frac{\pi}{1.1 + 
(d_{1e}^2 + d_{2e}^2)},$$ 
and for perfect GLI, we choose
$$\gamma_{\phi}(d_{je}) = \pi - \frac{\pi}{1.1 + \left(d_{je}\right)^2 }$$
such that both parameters are strictly increasing functions of the distances.
\footnote{A similar model for distance observation noise is used in \cite{Gezici}. 
Since $\phi_e \in [0,\pi]$, the angle observation error variance cannot diverge with distance, and we upper bounded the variance term by $\pi/1.1$. To avoid zero error variances at zero distance, we introduce a $0.1$ offset to numerator and denominator of $\gamma$ and $\gamma_{\phi}$, respectively.}
We consider node capability parameters $\rho_1=\rho_2=\rho_e=1$, unless stated otherwise.
The theoretical key rates in Section~\ref{sec:keyrates} converge as $n\to\infty$, therefore they are calculated for large enough $n$, using the forward algorithm procedure.
 
\subsubsection*{Results}

Due to computational limitations,
we consider examples in which $M\leq 11$, and $B \leq 3$. 
Note that this choice limits the maximum achievable secret key rate\footnote{For instance, 
for $B=1$, there are $13$ different possible distance combinations. 
 Consequently, a key rate of $\log 13$
is an absolute upper bound for no GLI even in the case when the eavesdropper does not obtain any observation.}.

\begin{figure*}[ht]
\begin{minipage}[b]{0.32\linewidth}
\hspace{-0.2cm}
\includegraphics[width=1\textwidth]{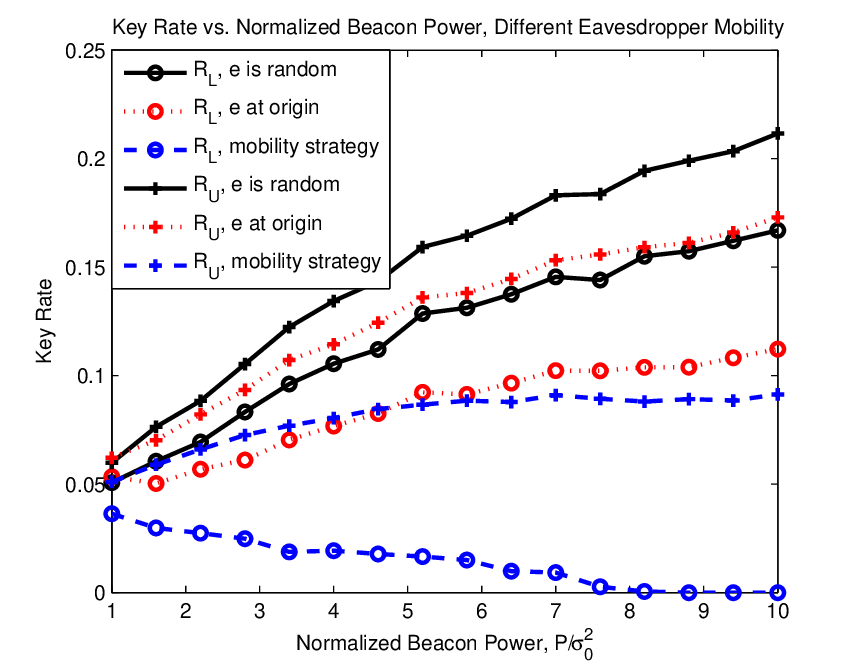}
\caption{Effect of eavesdropper mobility on the key rate}
\label{fig:emobility}
\end{minipage}
\hspace{0.1cm}
\begin{minipage}[b]{0.32\linewidth}
\hspace{-0.2cm}
\includegraphics[width=1.1\textwidth]{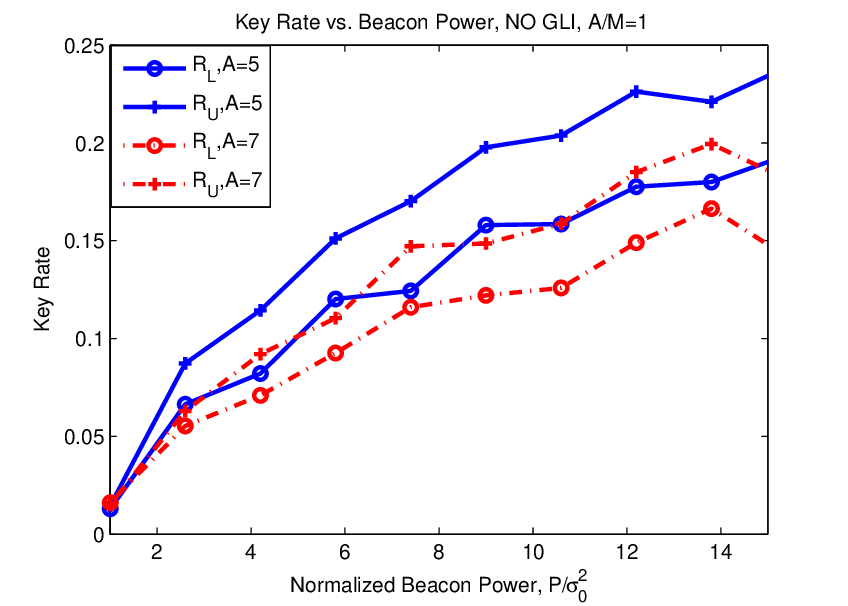}
\caption{Bounds for no GLI vs normalized beacon power, for different $M$, $B=1$, $A/M=1$}
\label{fig:nogli1}
\end{minipage}
\hspace{0.1cm}
\begin{minipage}[b]{0.32\linewidth}
\hspace{-0.3cm}
\includegraphics[width=1.1\textwidth]{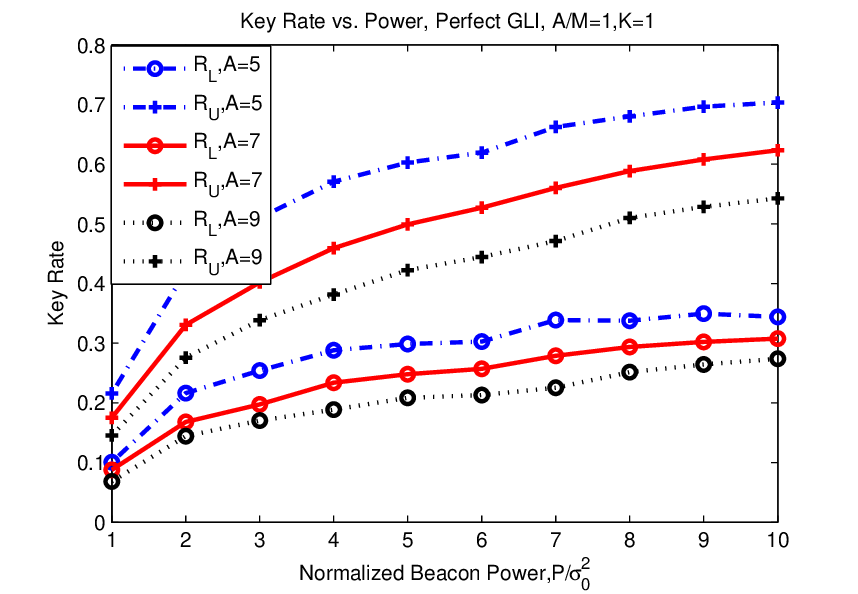}
\caption{Bounds for perfect GLI vs normalized beacon power, for different $M$, $B=1$, $A/M=1$}
\label{fig:gli1}
\end{minipage}
\end{figure*}

Then, we analyze the effect of the different grid size, field area and GLI on the theoretical key rates. 
In Figures~\ref{fig:nogli1} and \ref{fig:gli1}, we plot
the bounds on the achievable key rate with respect to the normalized beacon power $P/\sigma_0^2$ for different grid size $M$
 for no GLI and perfect GLI cases, respectively. We assumed a constant ratio of field size and grid size, $A/M = 1$, 
and considered $B=1$. We can see that, there is a diminishing return on the increased
power levels for the achievable key rate. Furthermore, we can see that increasing the field area $A^2$
has a negative impact on the key rate despite the increase 
in $M$, which is due to the fact that the common information of the legitimate nodes decreases as a result of increase in their observation
error variance. 
\begin{techreport}
Next, in Figures~\ref{fig:nogli2} and \ref{fig:gli2}, we plot the bounds with respect to beacon power $P$,
for different step size $B=1$ for no GLI and perfect GLI cases, respectively. We assumed $M=5$ for the no GLI case,
and $M=7$ for perfect GLI case, and in both cases, the ratio of field size and grid size is constant, such that $A/M =1$.
We can clearly see the positive effect of the increased step size on the secret key rate. 
This is due to the increase in different distance combinations that are possible.
\noindent \begin{figure}[ht]
\centerline{\includegraphics[width=7cm]{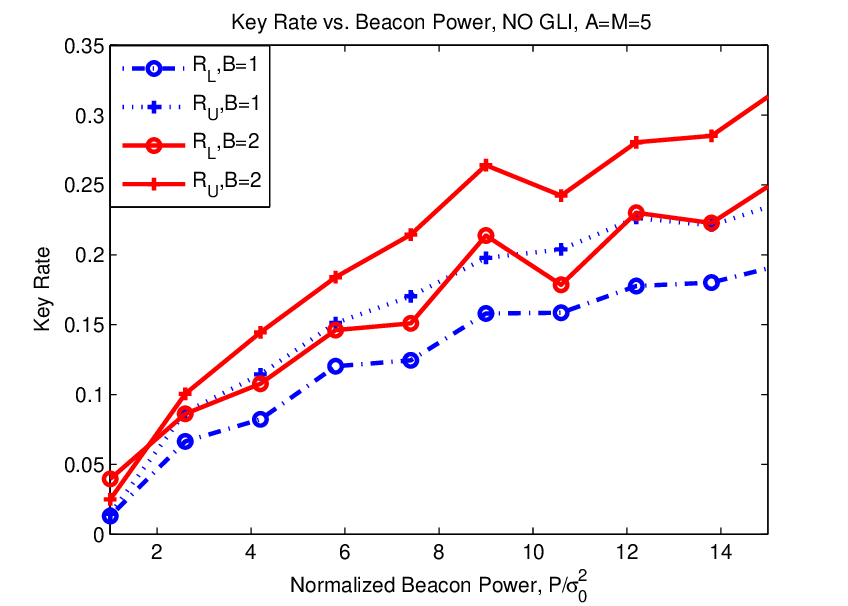}}
    \caption{Upper and lower bounds for no GLI vs normalized beacon power $P/\sigma_0^2$, $M=A=5$, for different $B$}
    \label{fig:nogli2}
\end{figure}
\noindent \begin{figure}[ht]
\centerline{\includegraphics[width=7cm]{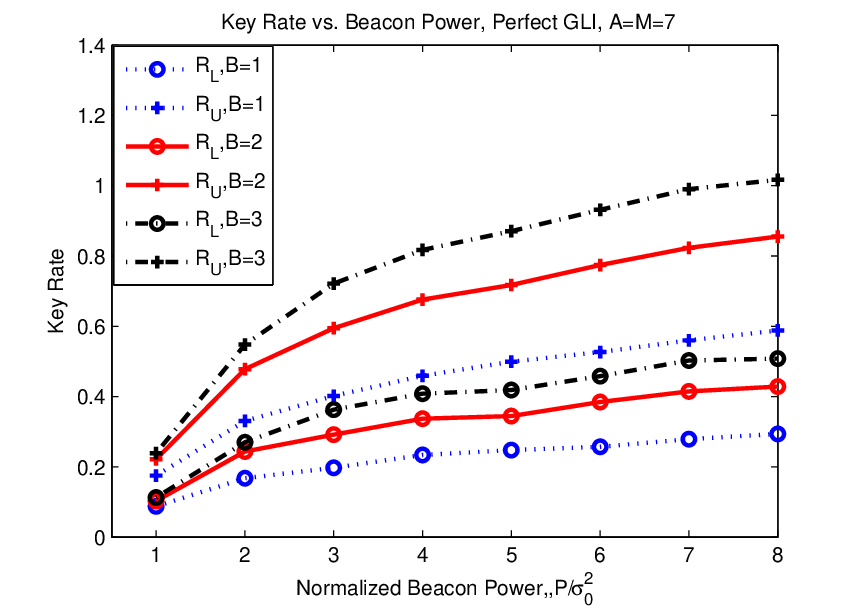}}
    \caption{Upper and lower bounds for perfect GLI vs normalized beacon power $P/\sigma_0^2$, $M=A=7$, for different $B$}
    \label{fig:gli2}
\end{figure}
\end{techreport}

Finally, we analyze the effect of eavesdropper mobility on the achievable key rate. In Figure~\ref{fig:emobility}, for $M=7$, 
$A=5$ and $B=1$, we plot the secret key rate bounds versus beacon power for the cases where
the eavesdropper
i) follows the random mobility pattern described in the setup with parameter $B=1$, ii) stays at the origin, and iii) follows the man in the middle strategy
described in Section~\ref{sec:attacker}, i.e., moves to the mid point of its location estimates
of nodes $1$ and $2$. We can see that, compared to following a random mobility pattern, the eavesdropper can reduce the 
achievable secret key rate significantly by following this strategy.
However, the rate still remains positive.
We observe that the eavesdropper can also reduce the key rate by simply staying static at a certain favorable location, rather than moving randomly. 
However, in practice this may not be feasible for the eavesdropper, since by staying put, it will lose connection completely with the legitimate nodes in a large region.


\section{Conclusion}\label{s:conclusion}
In this paper, we showed
that relative localization information is an additional resource for generating
secret key bits in mobile networks.
 We studied the information theoretic limits of secret key
generation, and characterized lower and upper bounds of key
rates utilizing results for the
cases in which the nodes are/are not capable of observing their global locations.
Focusing on the special case where the observation noise is i.i.d. Gaussian, 
we studied the beacon power asymptotics, and
observed that, interestingly, when the eavesdropper cannot observe
the angle information, the secret key rate grows \emph{unboundedly}. 
The following research directions can be further investigated 
1) theoretical performance analysis of secret key generation 
in large networks, taking into account the recent advances in
network information theoretic security, and 
2) security analysis of various adversarial models,  such as
active jamming attacks, or impersonation attacks in unauthenticated networks.

\begin{appendices}
\section{Proofs of Theorems in Section~\ref{sec:beacon_power}}\label{app:keyratebounds}
\subsection{Proof of Theorem~\ref{t:RUbound}}
We first provide three lemmas that will be useful when proving the theorem.
\begin{lemma}\label{l:variancebound1}
Let $x$ and $y$ be random variables. Then,
 \begin{align}
\var(x+y)   &\leq  2\var(x)+ 2\var(y). \label{RUbound3}
\end{align}
If $x$ and $y$ are independent, then
$\var(x+y)   = \var(x) + \var(y)$.
\end{lemma}
\begin{techreport}
\begin{proof}
When $x$ and $y$ are not independent,
\begin{align}
\var(x+y) & = \var(x) + \var(y) + 2 \cov(x,y) \nonumber\\
          &\leq \var(x) + \var(y) + 2\sqrt{\var(x)\var(y)} \nonumber\\
          &\leq 2\var(x)+ 2\var(y). \label{app:variancebound1_1}
\end{align}
where \eqref{app:variancebound1_1} follows from the fact that $\var(x) + \var(y) \geq 2\sqrt{\var(x)\var(y)}$. When $x$ and $y$ are independent,
$\cov(x,y) = 0$, implying the result.
\end{proof}
\end{techreport}
\begin{lemma}\label{l:variancebound3}
Let $x$ be a random variable such that $\expect[x] \geq \mu$, where $\mu > -1$. Let $\alpha = \frac{\sqrt{1+\mu}}{1+\mu}$. Then, $\var(\sqrt{[1+x]^+}) \leq \var(\alpha x)$.
\end{lemma}
\begin{techreport}
\begin{proof}
Assume $\expect[x] = \mu'$, where $\mu' \geq \mu$. Let $\alpha' = \frac{\sqrt{1+\mu'}}{1+\mu'}$.
Let us define $f_1(x) \triangleq \sqrt{[1+x]^+}$,
and $f_2(x) \triangleq \alpha'(1+x)$. 
\begin{singlecolumn}
Note that,
\begin{align*}
\var(\sqrt{[1+x]^+}) =   \var(f_1(x)) 
					 \leq \expect\big{[}(f_1(x) - \expect[f_2(x)])^2\big{]}
\end{align*}
since the centralized second moment is minimized around the mean.
Also,
$
\var(\alpha x) \geq \var(\alpha' x) 
			  =   \var(f_2(x))
$
since $\mu' \geq \mu$, and $\alpha' \leq \alpha$.  Therefore, it suffices to show that $\forall x$,
\begin{align}
|f_1(x) - \expect[f_2(x)]| \leq |f_2(x) - \expect[f_2(x)]| 
                           = |f_2(x) - \sqrt{1 + \mu'}|. \label{eq:lemmasqrtbound}
\end{align}
\end{singlecolumn}
\begin{multicolumn}
Note that,
\begin{align*}
\var(\sqrt{[1+x]^+}) =   \var(f_1(x)) 
					 \leq \expect\big{[}(f_1(x) - \expect[f_2(x)])^2\big{]}
\end{align*}
since the centralized second moment is minimized around the mean.
Also,
\begin{align*}
\var(\alpha x) \geq \var(\alpha' x) 
			  =   \var(f_2(x))
\end{align*}
since $\mu' \geq \mu$, and $\alpha' \leq \alpha$.  Therefore, it suffices to show that $\forall x$,
\begin{align}
|f_1(x) - \expect[f_2(x)]| &\leq |f_2(x) - \expect[f_2(x)]| \\
                           &= |f_2(x) - \sqrt{1 + \mu'}|. \label{eq:lemmasqrtbound}
\end{align}
\end{multicolumn}
i) First note that
$f_1(\mu') = f_2(\mu') = \sqrt{1+\mu'}$.
Therefore, the condition \eqref{eq:lemmasqrtbound} is satisfied for $x = \mu'$. \\
ii) For $x > \mu'$,
\begin{align}
f_1(x) &\leq f_1(\mu') + f_1'(\mu')(x-\mu')   \label{eq:lemmasqrt2}\\
       &\leq \sqrt{1+\mu'} + \alpha(x -\mu') = f_2(x) \nonumber
\end{align}
where $f_1'(\mu')$ is the first derivative of $f_1(x)$ at point $x = \mu'$. \eqref{eq:lemmasqrt2} follows from the fact that
$f_1(x)$ is a strictly concave function in the interval $[-1 ~ \infty)$.
Therefore, condition \eqref{eq:lemmasqrtbound} is satisfied for $x > \mu'$. \\
iii)  Combining the facts that $f_1(x)$ is a strictly concave
 function of $x$ in the interval $[-1 ~ \infty )$; $f_2(x)$
is linear; $f_1(-1) = f_2(-1)$;  and
$f_1(\mu') = f_2(\mu')$, we can see that $\sqrt{1+\mu'} > f_1(x) \geq f_2(x) $ when $ -1 < x < \mu'$. Therefore,
 condition \eqref{eq:lemmasqrtbound}
is satisfied for $ -1 < x < \mu'$.
iv) When $x< -1$, $f_1(x)=0$ and $f_2(x)<0$, therefore, condition \eqref{eq:lemmasqrtbound}
is satisfied.
This concludes the proof.
\end{proof}
\end{techreport}
\begin{lemma}\label{l:variancebound4}
Let $x$, $y$ be  random variables. Then,
\begin{align*}
\var\left( \expect\left[ 1 - \sqrt{(1+x)^+} ~ \big{|} y\right] \right)
    \leq \expect\left[ \big{(}\expect\big{[} ~|x|~ \big{|} ~ y \big{]}\big{)}^2 \right].
\end{align*}
\end{lemma}
\begin{techreport}
\begin{proof}
Note that
\begin{singlecolumn}
\begin{align*}
\var\left( \expect\left[ 1 - \sqrt{(1+x)^+} ~\big{|} y\right] \right)
  \leq \expect\left[\expect\left[ 1 - \sqrt{(1+x)^+} ~\big{|} y\right]^2 \right].
\end{align*}
\end{singlecolumn}
\begin{multicolumn}
\begin{align*}
&\var\left( \expect\left[ 1 - \sqrt{(1+x)^+} ~\big{|} y\right] \right) \\
& \quad \leq \expect\left[\expect\left[ 1 - \sqrt{(1+x)^+} ~\big{|} y\right]^2 \right].
\end{align*}
\end{multicolumn}
Since for any $x$, $|1 - \sqrt{(1+x)^+} | \leq |x|$,
\begin{align*}
 \bigg{|}\expect\left[ 1 - \sqrt{(1+x)^+} ~\big{|} y\right]\bigg{|}
\leq  \expect\big{[} ~|x|~ \big| ~ y \big{]}
\end{align*}
is satisfied for any $y$, which completes the proof.
\end{proof}
\end{techreport}
\begin{journal}
The proofs of Lemma~\ref{l:variancebound1}, \ref{l:variancebound3} and \ref{l:variancebound4} are provided in our technical report \cite{techreport}
due to space constraints.
\end{journal}
Now, we proceed as follows.
Assume without loss of generality that $\rho_{\min} = \min(\rho_1,\rho_2) = \rho_1$.
When $\hat{\phiv}_e \neq \emptyset$,
\begin{singlecolumn}
\begin{align}
R_U = & \lim_{n\to\infty}\frac{1}{n}  I(\hat{\dv}_1;\hat{\dv}_2| \hat{\dv}_{1e}, \hat{\dv}_{2e},\hat{\phiv}_e)  
    \leq \lim_{n\to\infty}\frac{1}{n} \big{(} h(\hat{\dv}_1|\hat{\dv}_{1e}, \hat{\dv}_{2e},\hat{\phiv}_e) - h(\hat{\dv}_1|\dv_{12}) \big{)}
 \label{app:RUboundMarkovd12}\\
    \leq & \lim_{n\to\infty}\frac{1}{n} \sum_{i=1}^n\bigg{(}  h(\hat{d}_1[i] | \hat{d}_{1e}[i], \hat{d}_{2e}[i], \hat{\phi}_e[i] ) - 
    h(\hat{d}_1[i] - d_{12}[i] | d_{12}[i])\bigg{)} \nonumber\\
       = & h(\hat{d}_1|\hat{d}_{1e},\hat{d}_{1e},\hat{\phi}_e) - h(\hat{d}_1 - d_{12}|d_{12})\label{app:RUboundmaineqn}
\end{align}
\end{singlecolumn}
\begin{multicolumn}
\begin{align}
R_U = & \lim_{n\to\infty}\frac{1}{n}  I(\hat{\dv}_1;\hat{\dv}_2| \hat{\dv}_{1e}, \hat{\dv}_{2e},\hat{\phiv}_e)  \nonumber\\
    \leq & \lim_{n\to\infty}\frac{1}{n} \big{(} h(\hat{\dv}_1|\hat{\dv}_{1e}, \hat{\dv}_{2e},\hat{\phiv}_e) - h(\hat{\dv}_1|\dv_{12}) \big{)}
 \label{app:RUboundMarkovd12}\\
    \leq & \lim_{n\to\infty}\frac{1}{n} \sum_{i=1}^n\bigg{(}  h(\hat{d}_1[i] | \hat{d}_{1e}[i], \hat{d}_{2e}[i], \hat{\phi}_e[i] ) - \nonumber\\
   & h(\hat{d}_1[i] - d_{12}[i] | d_{12}[i])\bigg{)} \nonumber\\
       = & h(\hat{d}_1|\hat{d}_{1e},\hat{d}_{1e},\hat{\phi}_e) - h(\hat{d}_1 - d_{12}|d_{12})\label{app:RUboundmaineqn}
\end{align}
\end{multicolumn}
where \eqref{app:RUboundMarkovd12} follows from the fact that $\hat{\dv}_1 \rightarrow \dv_{12}\rightarrow (\hat{\dv}_2,\hat{\dv}_{1e},\hat{\dv}_{2e},\hat{\phiv}_e)$ forms a Markov chain, and \eqref{app:RUboundmaineqn} follows from the fact that
all of the random variables $\hat{d}_1[i], \hat{d}_{1e}[i], \hat{d}_{2e}[i], \hat{\phi}_e[i]$ have a stationary distribution, denoted as
$\hat{d}_1, \hat{d}_{1e}, \hat{d}_{2e}$ and $\hat{\phi}_e$, respectively.
The second term in \eqref{app:RUboundmaineqn} can be found as
\begin{align}
h(\hat{d}_1 - d_{12}|d_{12}) = \frac{1}{2} \expect \left[ \log\left( \frac{\gamma(d_{12})\rho_1}{P} \right) \right] \label{entropyw1}
\end{align}
from the definition of $\hat{d}_1[i]$. Now, we bound the first term in \eqref{app:RUboundmaineqn}.
Let us define $$\hat{d}_e \triangleq \sqrt{[\hat{d}_{1e}^2 + \hat{d}_{2e}^2 - 2\hat{d}_{1e}\hat{d}_{2e}\cos(\hat{\phi}_e)]^+ }.$$ 
Then,
\begin{align}
h(\hat{d}_1 |  \hat{d}_{1e}, \hat{d}_{2e}, \hat{\phi}_e )
              & \leq h \big{(} \hat{d}_1 | \hat{d}_e  \big{)} 
               \leq h \big{(} \hat{d}_1 - \hat{d}_e  \big{)}.  \label{RUbound1}
\end{align}
Note that for a given variance, Gaussian distribution maximizes the entropy.
Therefore, the entropy of a Gaussian random variable that has a variance
identical to that of $\hat{d}_1-\hat{d}_e$ will be an upper bound for
\eqref{RUbound1}. We proceed as follows.
\begin{singlecolumn}
\begin{align}
\var\big{(}  \hat{d}_1 - \hat{d}_e  \big{)} =
  \expect\big{[} \var\big{(} \hat{d}_1 - \hat{d}_e |d_{12},d_{1e},d_{2e} \big{)}\big{]} 
 +\var \big{(} \expect\big{[}  \hat{d}_1 - \hat{d}_e |   d_{12},d_{1e},d_{2e}  \big{]} \big{)}\label{app:variance},
\end{align}
\end{singlecolumn}
\begin{multicolumn}
\begin{align}
\var\big{(}  \hat{d}_1 - \hat{d}_e  \big{)} =
&  \expect\big{[} \var\big{(} \hat{d}_1 - \hat{d}_e |d_{12},d_{1e},d_{2e} \big{)}\big{]}  \nonumber\\
&  +\var \big{(} \expect\big{[}  \hat{d}_1 - \hat{d}_e |   d_{12},d_{1e},d_{2e}  \big{]} \big{)}\label{app:variance},
\end{align}
\end{multicolumn}
where \eqref{app:variance} follows from the fact that for any dependent random variables $x$ and $y$,
$\var(x) = \expect[\var(x|y)]+ \var(\expect[x|y])$.
We now find an upper bound on the first term of \eqref{app:variance}. Note that,
\begin{singlecolumn}
\begin{align}
 \var\big{(} \hat{d}_1 - \hat{d}_e |d_{12},d_{1e},d_{2e} \big{)}  =  \var(\hat{d}_1|d_{12})  
		 + \var\big{(} \hat{d}_e | d_{12},d_{1e},d_{2e} \big{)}  \label{RUbound2}
\end{align}
\end{singlecolumn}
\begin{multicolumn}
\begin{align}
 \var\big{(} \hat{d}_1 - \hat{d}_e |d_{12},d_{1e},d_{2e} \big{)}  = & \var(\hat{d}_1|d_{12}) \nonumber\\ 
		& + \var\big{(} \hat{d}_e | d_{12},d_{1e},d_{2e} \big{)}  \label{RUbound2}
\end{align}
\end{multicolumn}
due to Lemma~\ref{l:variancebound1}, since 
$\hat{d}_1 \rightarrow (d_{12},d_{1e},d_{2e})\rightarrow(\hat{d}_{1e},\hat{d}_{2e},\hat{\phi}_e)\rightarrow \hat{d}_e$
forms a Markov chain, and the fact that $\hat{d}_1$ is independent of ${d}_{1e},d_{2e}$ given $d_{12}$.
The first term in \eqref{RUbound2} is equal to
\begin{align}
\var(\hat{d}_1 | d_{12})     = \var\big{(}w_1 | d_{12} \big{)}  = \frac{\gamma(d_{12})\rho_1}{P}. 
\label{app:varhatd1}
\end{align}
We bound the second term in \eqref{RUbound2} as follows. Let us define 
\begin{singlecolumn}
\begin{align*}
\kappa &\triangleq \frac{1}{d_{12}^2} \bigg{(} 2(d_{1e}-d_{2e}\cos(\hat{\phi_e})) w_{1e} + {w_{1e}}^2 
  + {w_{2e}}^2 + 2(d_{2e}-d_{1e}\cos(\hat{\phi}_e)) w_{2e} \nonumber\\
 & ~ + 2d_{1e}d_{2e} (\cos(\phi_e) - \cos(\hat{\phi}_e))  
  -  2 {w_{1e}}{w_{2e}} \cos(\hat{\phi}_e)\bigg{)}.
\end{align*}
\end{singlecolumn}
\begin{multicolumn}
\begin{align*}
\kappa &\triangleq \frac{1}{d_{12}^2} \bigg{(} 2(d_{1e}-d_{2e}\cos(\hat{\phi_e})) w_{1e} + {w_{1e}}^2 \nonumber\\
 & ~ + {w_{2e}}^2 + 2(d_{2e}-d_{1e}\cos(\hat{\phi}_e)) w_{2e} \nonumber\\
 & ~ + 2d_{1e}d_{2e} (\cos(\phi_e) - \cos(\hat{\phi}_e))  
  -  2 {w_{1e}}{w_{2e}} \cos(\hat{\phi}_e)\bigg{)}.
\end{align*}
\end{multicolumn}
Then,
\begin{singlecolumn}
\begin{align}
 \var ( \hat{d}_e | d_{12},d_{1e},d_{2e} ) 
 &= \var\bigg{\{} \bigg{(} \bigg{[}d_{1e}^2 +d_{2e}^2
           + 2(d_{1e}-d_{2e}\cos\hat{\phi_e}) w_{1e} +
2(d_{2e}-d_{1e}\cos\hat{\phi}_e) w_{2e} \nonumber\\
 &   ~ - 2d_{1e}d_{2e} \cos\hat{\phi}_e + {w_{1e}}^2 
  + {w_{2e}}^2 -  2 {w_{1e}}{w_{2e}} \cos\hat{\phi}_e
\bigg{]}^+ \bigg{)}^{0.5}  | d_{12},d_{1e},d_{2e}  \bigg{\}}  \label{RUeqn1}\\
 & \leq d_{12}^2\var (\sqrt{[1 + \kappa]^+}), \label{RUeqn2}
\end{align}
\end{singlecolumn}
\begin{multicolumn}
\begin{align}
& \quad \var ( \hat{d}_e | d_{12},d_{1e},d_{2e} )  \nonumber\\
 &= \var\bigg{\{} \bigg{(} \bigg{[}d_{1e}^2 +d_{2e}^2
           + 2(d_{1e}-d_{2e}\cos\hat{\phi_e}) w_{1e} + \nonumber\\
 & ~   2(d_{2e}-d_{1e}\cos\hat{\phi}_e) w_{2e}- 2d_{1e}d_{2e} \cos\hat{\phi}_e + {w_{1e}}^2 \nonumber\\
 & ~   + {w_{2e}}^2 -  2 {w_{1e}}{w_{2e}} \cos\hat{\phi}_e
\bigg{]}^+ \bigg{)}^{0.5}  | d_{12},d_{1e},d_{2e}  \bigg{\}}  \label{RUeqn1}\\
 & \leq d_{12}^2\var (\sqrt{[1 + \kappa]^+}), \label{RUeqn2}
\end{align}
\end{multicolumn}
where \eqref{RUeqn1} follows due to the definitions of $\hat{d}_{1e}$, $\hat{d}_{2e}$ and $\hat{\phi}_{e}$.
\eqref{RUeqn2} follows due to definition of $\kappa$, and the cosine law $d_{12}^2 = d_{1e}^2 + d_{2e}^2 - 2d_{1e}d_{2e}\cos(\phi_e)$.
Now we will apply Lemma~\ref{l:variancebound3} to bound \eqref{RUeqn2}. First, note that
\begin{align}
& \expect[\kappa] =\frac{1}{d_{12}^2}\expect\bigg{[}w_{1e}^2 + w_{2e}^2 - 2d_{1e}d_{2e} (\cos\phi_e - \cos\hat{\phi}_e)\bigg{]} \nonumber\\
           &\geq \frac{1}{P d_{12}^2}\expect\bigg{[} w_{1e}^2 + w_{2e}^2 - 2d_{1e}d_{2e}|w_{\phi_e}| \bigg{]} \nonumber\\
           &\geq \frac{\rho_e}{P d_{12}^2}\expect\bigg{[} \gamma(d_{1e}) + \gamma(d_{2e}) 
- 2d_{1e}d_{2e}\sqrt{\frac{2 P \gamma_{\phi}(d_{1e},d_{2e})}{\pi\rho_e}} \bigg{]}, \label{app:halfnormal1}
\end{align}
where \eqref{app:halfnormal1} follows from the fact that since $w_{\phi_e}$ is zero mean Gaussian,
$|w_{\phi_e}|$ follows a Half-normal distribution
with $\expect(|w_{\phi_e}|) = \sqrt{\frac{2 \gamma_{\phi}(d_{1e},d_{2e})\rho_e}{P\pi}}$.
We can choose $P_1$ such that for any beacon power $P > P_1$, $\expect[x]> -\frac{3}{4}$. Let $\mu = -\frac{3}{4}$, and $\alpha = \frac{\sqrt{1+\mu}}{1+\mu}=2$.
Due to Lemma~\ref{l:variancebound3}, we obtain
\begin{singlecolumn}
\begin{align}
 d_{12}^2\var &(\sqrt{[1+\kappa]^+})\leq  d_{12}^2\var(\alpha \kappa) \nonumber\\
       &\leq  \frac{4\alpha^2}{d_{12}^2}\bigg{(} \var(2(d_{1e}-d_{2e}\cos(\hat{\phi_e})) w_{1e}) 
          + \var(2(d_{2e}-d_{1e}\cos(\hat{\phi}_e)) w_{2e}) \label{RUbound_variance_expand}\\
       & + \var(2d_{1e}d_{2e} (\cos(\phi_e) - \cos(\hat{\phi}_e))) 
         + \var(2 {w_{1e}}{w_{2e}} \cos(\hat{\phi}_e))
         + \var(w_{1e}^2) + \var(w_{2e}^2) \bigg{)} \nonumber\\
        \leq& \expect\bigg{[}\frac{16 \rho_e}{P d_{12}^2}\bigg{(} 4(d_{1e}+d_{2e})^2 (\gamma(d_{1e})+\gamma(d_{2e})) 
         + 4(d_{1e}d_{2e})^2 \gamma_{\phi}(d_{1e},d_{2e}) + o\left(\frac{1}{P}\right) \bigg{)}\bigg{]}, \label{RUbound_variance_expand3}
\end{align}
\end{singlecolumn}
\begin{multicolumn}
\begin{align}
d_{12}^2&\var(\sqrt{[1+\kappa]^+}) \leq  d_{12}^2\var(\alpha \kappa) \nonumber\\
                             \leq&  \frac{4\alpha^2}{d_{12}^2}\bigg{(} \var(2(d_{1e}-d_{2e}\cos(\hat{\phi_e})) w_{1e}) \label{RUbound_variance_expand}\\
                        & + \var(2(d_{2e}-d_{1e}\cos(\hat{\phi}_e)) w_{2e}) \nonumber\\
                        & + \var(2d_{1e}d_{2e} (\cos(\phi_e) - \cos(\hat{\phi}_e))) \nonumber\\
                        & + \var(2 {w_{1e}}{w_{2e}} \cos(\hat{\phi}_e))
                         + \var(w_{1e}^2) + \var(w_{2e}^2) \bigg{)} \nonumber\\
                             \leq& \expect\bigg{[}\frac{16 \rho_e}{P d_{12}^2}\bigg{(} 4(d_{1e}+d_{2e})^2 (\gamma(d_{1e})+\gamma(d_{2e})) \\
                        & + 4(d_{1e}d_{2e})^2 \gamma_{\phi}(d_{1e},d_{2e}) + o\left(\frac{1}{P}\right) \bigg{)}\bigg{]}, \label{RUbound_variance_expand3}
\end{align}
\end{multicolumn}
where \eqref{RUbound_variance_expand} follows from applying Lemma~\ref{l:variancebound1} to $\var(\kappa)$ twice, and
\eqref{RUbound_variance_expand3} follows from the fact that  
$$\var(2(d_{ie}-d_{je}\cos(\hat{\phi}_e)) w_{ie}) \leq \var(2(d_{ie}+d_{je}) w_{ie})$$ for $i,j \in \{1,2\}$,
and 
\begin{singlecolumn}
\begin{align*}
\var&(2d_{1e}d_{2e} (\cos(\phi_e) - \cos(\hat{\phi}_e))) 
   =  \var(2d_{1e}d_{2e} (\cos(\phi_e) - \cos(\phi_e + w_{\phi} )) 
	\leq  \var(2d_{1e}d_{2e} w_{\phi}).
\end{align*}
\end{singlecolumn}
\begin{multicolumn}
\begin{align*}
\var&(2d_{1e}d_{2e} (\cos(\phi_e) - \cos(\hat{\phi}_e))) \\
	&=  \var(2d_{1e}d_{2e} (\cos(\phi_e) - \cos(\phi_e + w_{\phi} )) \\
	&\leq  \var(2d_{1e}d_{2e} w_{\phi}).
\end{align*}
\end{multicolumn}
Now, we upper bound the second term of \eqref{app:variance} as 
\begin{singlecolumn}
\begin{align}
&\var \big{(} \expect \big{[} \hat{d}_1 - \hat{d}_e | d_{12},d_{1e},d_{2e} \big{]} \big{)}   
   = \var \bigg{(} d_{12}\expect \bigg{[} 1- \sqrt{\left(1 + \kappa \right)^+ } 
     | d_{12},d_{1e},d_{2e} \bigg{]} \bigg{)} \label{app:expectation1} \\
   &\leq \expect\left[ d_{12}^2 \left( \expect\big{[}|x|~|d_{12},d_{1e},d_{2e} \big{]}\right)^2\right] \label{app:expectation2}\\
&\leq \expect \bigg{[} \frac{1}{d_{12}^2} \expect\big{[}
  2(d_{1e} + d_{2e})(|w_{1e}|+|w_{2e}|) + w_{1e}^2 + w_{2e}^2 
  + 2d_{1e}d_{2e}|w_{\phi_e}| + 2 |w_{1e}w_{2e}|~ | d_{12},d_{1e},
  d_{2e}\big{]}^2 \bigg{]} \nonumber\\
& = \expect\bigg{[}\frac{\rho_e^2}{Pd_{12}^2} \bigg{(}2(d_{1e}+d_{2e})
\frac{\sqrt{\gamma(d_{1e})}+\sqrt{\gamma(d_{2e})}}{\sqrt{\rho_e}} 
   + 2d_{1e}d_{2e}\gamma_{\phi}(d_{1e},d_{2e}) \nonumber\\
   & ~ ~ + \frac{\gamma(d_{1e})+\gamma(d_{2e})}{\sqrt{P}} + 2\frac{\sqrt{\gamma(d_{1e})+\gamma(d_{2e}) }}{\sqrt{P\rho_e}} \bigg{)}^2 \bigg{]} \label{app:expectation3} \\
& =\expect\bigg{[}\frac{\rho_e^2}{Pd_{12}^2} \bigg{(}\frac{4(d_{1e}+d_{2e})^2}{\rho_e}
(\sqrt{\gamma(d_{1e})}+\sqrt{\gamma(d_{2e})})^2 
  + 4(d_{1e}d_{2e}\gamma_{\phi}(d_{1e},d_{2e}))^2 \nonumber\\
& ~ ~  + \frac{8(d_{1e}+d_{2e})d_{1e}d_{2e}}{\rho_e}\big{(}\sqrt{\gamma(d_{1e})} 
\quad +\sqrt{\gamma(d_{2e})} \big{)}\sqrt{\gamma_{\phi}(d_{1e},d_{2e})}
\bigg{)}\bigg{]} + o(1/P) \label{RUeqn6}.
\end{align}
\end{singlecolumn}
\begin{multicolumn}
\begin{align}
&\var \big{(} \expect \big{[} \hat{d}_1 - \hat{d}_e | d_{12},d_{1e},d_{2e} \big{]} \big{)}    \nonumber\\
&= \var \bigg{(} d_{12}\expect \bigg{[} 1- \sqrt{\left(1 + \kappa \right)^+ } | d_{12},d_{1e},d_{2e} \bigg{]} \bigg{)} \label{app:expectation1} \\
&\leq \expect\left[ d_{12}^2 \left( \expect\big{[}|x|~|d_{12},d_{1e},d_{2e} \big{]}\right)^2\right] \label{app:expectation2}\\
&\leq \expect \bigg{[} \frac{1}{d_{12}^2} \expect\big{[}
2(d_{1e} + d_{2e})(|w_{1e}|+|w_{2e}|) + w_{1e}^2 + w_{2e}^2 \nonumber\\
& + 2d_{1e}d_{2e}|w_{\phi_e}| + 2 |w_{1e}w_{2e}|~ | d_{12},d_{1e},
d_{2e}\big{]}^2 \bigg{]} \nonumber\\
& = \expect\bigg{[}\frac{\rho_e^2}{Pd_{12}^2} \bigg{(}2(d_{1e}+d_{2e})
\frac{\sqrt{\gamma(d_{1e})}+\sqrt{\gamma(d_{2e})}}{\sqrt{\rho_e}} \nonumber\\
&+ 2d_{1e}d_{2e}\gamma_{\phi}(d_{1e},d_{2e}) \nonumber\\
&+ \frac{\gamma(d_{1e})+\gamma(d_{2e})}{\sqrt{P}} + 2\frac{\sqrt{\gamma(d_{1e})+\gamma(d_{2e}) }}{\sqrt{P\rho_e}} \bigg{)}^2 \bigg{]} \label{app:expectation3} \\
& =\expect\bigg{[}\frac{\rho_e^2}{Pd_{12}^2} \bigg{(}\frac{4(d_{1e}+d_{2e})^2}{\rho_e}
(\sqrt{\gamma(d_{1e})}+\sqrt{\gamma(d_{2e})})^2 \nonumber\\
& + 4(d_{1e}d_{2e}\gamma_{\phi}(d_{1e},d_{2e}))^2 \nonumber\\
&  + \frac{8(d_{1e}+d_{2e})d_{1e}d_{2e}}{\rho_e}\big{(}\sqrt{\gamma(d_{1e})} \nonumber\\
& \quad +\sqrt{\gamma(d_{2e})} \big{)}\sqrt{\gamma_{\phi}(d_{1e},d_{2e})}
\bigg{)}\bigg{]} + o(1/P) \label{RUeqn6}.
\end{align}
\end{multicolumn}
where \eqref{app:expectation1} follows from the fact that $\hat{d}_e = d_{12}\sqrt{(1+\kappa)^+}$, and
\eqref{app:expectation2} follows from Lemma~\ref{l:variancebound4}.
Finally,  we obtain
\begin{singlecolumn}
\begin{align}
 R_U &\leq h(\hat{d}_1 - \hat{d}_e) - h(\hat{d}_1 - d_{12} | d_{12}) \label{app:RUfinaleqn0}\\
	& \leq \frac{1}{2}\log\left(2\pi \var(\hat{d}_1-\hat{d}_e)\right)
- h(w_{1}|d_{12}) \label{app:RUfinaleqn1}\\
&= \frac{1}{2} \log\bigg{\{}
  2\pi\expect \bigg{[} \frac{\rho_e}{d_{12}^2P} \bigg{(}
 \frac{d_{12}^2\rho_1}{\rho_e} \gamma(d_{12}) 
 +  4(d_{1e}+d_{2e})^2 (\sqrt{\gamma(d_{1e})}
+\sqrt{\gamma(d_{2e})})^2 \nonumber\\
& + (4 d_{1e}d_{2e}+ 64(d_{1e}d_{2e})^2)\gamma_{\phi}(d_{1e},d_{2e}) 
    + 8(d_{1e}+d_{2e})d_{1e}d_{2e}\big{(}\sqrt{\gamma(d_{1e})} 
   +\sqrt{\gamma(d_{2e})} \big{)}\sqrt{\gamma_{\phi}(d_{1e},d_{2e})}
 \nonumber\\
&+ 64 (d_{1e}+d_{2e})^2 (\gamma(d_{1e})+\gamma(d_{2e})) \bigg{)}  + o\left(\frac{1}{P}\right) \bigg{]} \bigg{\}}  - \frac{1}{2}\expect\bigg{[}\log\left(\frac{2\pi\rho_1\gamma( d_{12})}{P}\right) \bigg{]}
\label{app:RUfinaleqn2},
 \end{align}
\end{singlecolumn}
\begin{multicolumn}
\begin{align}
 R_U &\leq h(\hat{d}_1 - \hat{d}_e) - h(\hat{d}_1 - d_{12} | d_{12}) \label{app:RUfinaleqn0}\\
	& \leq \frac{1}{2}\log\left(2\pi \var(\hat{d}_1-\hat{d}_e)\right)
- h(w_{1}|d_{12}) \label{app:RUfinaleqn1}\\
&= \frac{1}{2} \log\bigg{\{}
  2\pi\expect \bigg{[} \frac{\rho_e}{d_{12}^2P} \bigg{(}
 \frac{d_{12}^2\rho_1}{\rho_e} \gamma(d_{12}) \nonumber\\
& +  4(d_{1e}+d_{2e})^2 (\sqrt{\gamma(d_{1e})}
+\sqrt{\gamma(d_{2e})})^2 \nonumber\\
& + (4 d_{1e}d_{2e}+ 64(d_{1e}d_{2e})^2)\gamma_{\phi}(d_{1e},d_{2e}) \nonumber\\
&  + 8(d_{1e}+d_{2e})d_{1e}d_{2e}\big{(}\sqrt{\gamma(d_{1e})} \nonumber\\
&  +\sqrt{\gamma(d_{2e})} \big{)}\sqrt{\gamma_{\phi}(d_{1e},d_{2e})}
 \nonumber\\
&+ 64 (d_{1e}+d_{2e})^2 (\gamma(d_{1e})+\gamma(d_{2e})) \bigg{)}  + o\left(\frac{1}{P}\right) \bigg{]} \bigg{\}} \nonumber\\
& - \frac{1}{2}\expect\bigg{[}\log\left(\frac{2\pi\rho_1\gamma( d_{12})}{P}\right) \bigg{]}
\label{app:RUfinaleqn2},
 \end{align}
\end{multicolumn}
where 
\eqref{app:RUfinaleqn0} follows from \eqref{app:RUboundmaineqn} and \eqref{RUbound1},
\eqref{app:RUfinaleqn1} follows from the fact that entropy of $\hat{d}_1 - \hat{d}_e$ is upper bounded by
the entropy of a Gaussian random variable that has the same variance as $\hat{d}_1 - \hat{d}_e$.
The first term of \eqref{app:RUfinaleqn2} is obtained by combining combining
\eqref{app:varhatd1}, \eqref{RUbound_variance_expand3} and \eqref{RUeqn6}, and the second term of
\eqref{app:RUfinaleqn2} follows from \eqref{entropyw1}. As $P \to \infty$, the $P$
terms in \eqref{app:RUfinaleqn2} cancel each other since for any random variables $u$ and $v$,
\begin{align*}
\lim_{P\to\infty}\log\expect\left[ \frac{u}{P} + o\bigg{(}\frac{1}{P} \bigg{)} \right] - \expect\left[\log \frac{v}{P} \right]
 = \log\expect[u] - \expect[\log v]
\end{align*}
hence $\lim_{P\to\infty}R_U < \infty$.

\subsection{Proof of Theorem~\ref{t:RLbound}}

When the eavesdropper does not observe the angle, $\hat{\phiv}_e = \emptyset$. Hence
\begin{align}
R_L = \lim_{n\to\infty}\frac{1}{n} \left( h(\hat{\dv}_1|\hat{\dv}_{1e},\hat{\dv}_{2e})
                                         - h(\hat{\dv}_{1}|\hat{\dv }_{2}) \right)
                                         \label{app:RLeqn}
\end{align}
First, we show that the first term in \eqref{app:RLeqn} is finite.
\begin{singlecolumn}
\begin{align}
&\lim_{P, n \to \infty}\frac{1}{n}h(\hat{\dv}_{1}|\hat{\dv}_{1e},\hat{\dv}_{2e})
	 = \lim_{n\to\infty}\frac{1}{n}h(\dv_{12}|\dv_{1e},\dv_{2e}) 
	= \lim_{n\to\infty}\frac{1}{n}\sum_{i=1}^n h(d_{12}[i] | \dv_{1e},\dv_{2e},\{d_{12}[j]\}_{j=1}^{i-1}) \nonumber\\
    &=\lim_{n\to\infty}\frac{1}{n} \sum_{i=1}^n h(d_{12}[i] | \dv_{1e},\dv_{2e},\{\phi_e[j]\}_{j=1}^{i-1}) \label{app:RLbound1}\\
    &= \lim_{n\to\infty}\frac{1}{n}\sum_{i=1}^n h\bigg{(}
\big{(}d_{1e}[i]^2+d_{2e}[i]^2- 
2d_{1e}[i]d_{2e}[i]\cos(\phi_e[i])\big)^{0.5}|\dv_{1e},\dv_{2e},\{\phi_e[j]\}_{j=1}^{i-1}\bigg{)} 
 > -\infty  \label{app:RLbound3}
\end{align}
\end{singlecolumn}
\begin{multicolumn}
\begin{align}
&\lim_{P, n \to \infty}\frac{1}{n}h(\hat{\dv}_{1}|\hat{\dv}_{1e},\hat{\dv}_{2e})
	 = \lim_{n\to\infty}\frac{1}{n}h(\dv_{12}|\dv_{1e},\dv_{2e})  \nonumber\\
	&= \lim_{n\to\infty}\frac{1}{n}\sum_{i=1}^n h(d_{12}[i] | \dv_{1e},\dv_{2e},\{d_{12}[j]\}_{j=1}^{i-1}) \nonumber\\
    &=\lim_{n\to\infty}\frac{1}{n} \sum_{i=1}^n h(d_{12}[i] | \dv_{1e},\dv_{2e},\{\phi_e[j]\}_{j=1}^{i-1}) \label{app:RLbound1}\\
    &= \lim_{n\to\infty}\frac{1}{n}\sum_{i=1}^n h\bigg{(}
\big{(}d_{1e}[i]^2+d_{2e}[i]^2- \nonumber\\
&2d_{1e}[i]d_{2e}[i]\cos(\phi_e[i])\big)^{0.5}|\dv_{1e},\dv_{2e},\{\phi_e[j]\}_{j=1}^{i-1}\bigg{)}
 > -\infty  \label{app:RLbound3}
\end{align}
\end{multicolumn}
where \eqref{app:RLbound1} follows from the fact that a triangle is completely characterized by either three sides $(d_{12}[i]\,d_{1e}[i],d_{2e}[i])$, or two sides and an angle $(d_{12}[i]\,d_{1e}[i],\phi_e[i])$.
 Equation \eqref{app:RLbound3} follow from the cosine law.
Since the probability density function of $\phiv_e, \dv_{12}$ and $\dv_{1e},\dv_{2e}$ are well defined,
we can see that $h(\phi_e[i]|\dv_{1e},\dv_{2e},\{\phi_e[j]\}_{j=1}^{i-1})>-\infty$.
The second term
\begin{align}
\frac{1}{n} h(\hat{\dv}_1|\hat{\dv}_2)
         &\leq \frac{1}{n}h(\hat{\dv}_1 - \hat{\dv}_2) 
         = h(w_1 - w_2) \label{app:RLbound3.5}\\
         &\leq \frac{1}{2} \log\left(  2\pi \expect\left[ \frac{4\rho_{\max}\gamma(d_{12})}{P} \right] \right) \label{app:RLbound4} \\
         &=\frac{1}{2} \log\big{(} 2\pi \expect\left[ 4 \rho_{\max}\gamma(d_{12}) \right] \big{)}
         - \frac{1}{2} \log(P). \nonumber
\end{align}
where \eqref{app:RLbound3.5} follows due to the fact that conditioning reduces entropy, \eqref{app:RLbound4} follows from the fact 
$\hat{d}_{1}[i]-\hat{d}_2[i] = w_{1}[i] - w_{2}[i]$ for any $i$. Dropping the index $i$, we can see 
that $h(w_1 - w_2)$ is upper bounded by entropy of a Gaussian random variable that has the same variance as $w_1 - w_2$, which is
\begin{align*}
\var(w_1 - w_2) &= \var(\expect[w_1 - w_2]) + \expect(\var[w_1 - w_2]) \\
				&= \expect_{d_{12}}(\var[w_1 - w_2])
				\leq \expect\left[ \frac{4\rho_{\max}\gamma(d_{12})}{P} \right]. 
\end{align*}
 Therefore, we can see that
$
\lim_{P \to \infty} \frac{R_L}{\frac{1}{2} \log (P) }  \geq 1$
Now, we find an upper bound on $R_U$.
Note that
\begin{singlecolumn}
\begin{align}
R_U =  \lim_{n\to\infty}\frac{1}{n}  I(\hat{\dv}_1;\hat{\dv}_2| \hat{\dv}_{1e}, \hat{\dv}_{2e})  
    \leq  \lim_{n\to\infty}\frac{1}{n} \big{(} h(\hat{\dv}_1|\hat{\dv}_{1e},
 \hat{\dv}_{2e}) - h(\hat{\dv}_1|\dv_{12}) \label{app:RUeqn}
\end{align}
\end{singlecolumn}
\begin{multicolumn}
\begin{align}
R_U = & \lim_{n\to\infty}\frac{1}{n}  I(\hat{\dv}_1;\hat{\dv}_2| \hat{\dv}_{1e}, \hat{\dv}_{2e})  \nonumber\\
    \leq & \lim_{n\to\infty}\frac{1}{n} \big{(} h(\hat{\dv}_1|\hat{\dv}_{1e},
 \hat{\dv}_{2e}) - h(\hat{\dv}_1|\dv_{12}) \label{app:RUeqn}
\end{align}
\end{multicolumn}
where the first term of \eqref{app:RUeqn} is finite.
The second term can be upper bounded as
\begin{singlecolumn}
\begin{align*}
\frac{1}{n} h(\hat{\dv}_1|{\dv}_{12})
         = h(w_1|d_{12}) 
         \geq \expect\left[\frac{1}{2}\log\left( \frac{2\pi\rho_1 \gamma(d_{12})}{P} \right) \right] 
         =\frac{1}{2}\expect\left[\log\left(2\pi\rho_1\gamma(d_{12})\right) \right] -\frac{1}{2}\log(P),
\end{align*}
\end{singlecolumn}
\begin{multicolumn}
\begin{align*}
\frac{1}{n} h(\hat{\dv}_1|{\dv}_{12})
         & = h(w_1|d_{12}) \\
         &\geq \expect\left[\frac{1}{2}\log\left( \frac{2\pi\rho_1 \gamma(d_{12})}{P} \right) \right] \\
         &=\frac{1}{2}\expect\left[\log\left(2\pi\rho_1\gamma(d_{12})\right) \right] -\frac{1}{2}\log(P),
\end{align*}
\end{multicolumn}
therefore
$\lim_{P \to \infty} \frac{R_U}{\frac{1}{2} \log (P) }  \leq 1$.
Since $R_L \leq R_U$ by definition, the proof is complete.
\begin{techreport}
\section{On Observation Bias}\label{app:practical}
Nodes' observations may have bias due to several factors. We consider
two different source of bias; clock mismatch and multipath fading. We will see that
different types of bias may have different outcomes.
In this part, we present our results for no GLI. However, the conclusions are valid for perfect GLI as well.
\subsection{Clock Mismatch}
Assume that there is a clock mismatch between nodes $1$, and $2$. Consequently,
all the observations of $d_{12}$ of nodes $1$ and $2$ in localization phase
are shifted by a random value $\eta_1$ and $\eta_2$ respectively:
\begin{equation}
\hat{d}_j[i] = d_{12}[i]+ w_j[i] + \eta_j
\end{equation}
for $j\in \{1,2\}$, where $w_j[i]$ is as given in \eqref{eq:legit_obs}.
We assume the amount of clock mismatch is a non-random, but unknown parameter, which remains constant througout the entire session\footnote{The underlying assumption is that, the clock mismatch variations are much slower than the duration of the key-generation sessions}. Since $w_j[i]$, $j \in\{1,2\}$ are zero mean random
variables,
\begin{equation}
\expect[\hat{d}_j[i]|\eta_1,\eta_2] = \expect[d_{12}[i]] + \eta_j .
\end{equation}
Hence, with the knowledge of the statistics of the mobility, each node
$j \in \{1,2\}$ can obtain a perfect estimate of the amount of clock mismatch
$\eta_j$ as $n \to \infty$, by simply calculating the difference between the long-term average of the distance observations $\hat{d}_j[i]$ for all $i$ and the known mean distance $\expect[d_{12}]$. Then this value can be broadcast in the public discussion phase.
Therefore, clock mismatch does not affect the theoretical bounds of secret key
generation rates.
\subsection{Multipath Fading}
There may be a bias in the observations when the nodes experience multipath fading. An example
of this is time of arrival observation of distances when the nodes are not within their line of sight.
Note that, this kind of bias does not remain constant, and varies from one slot to the other. 
The impact of fading can be viewed as that of an additional observation noise source and the distance observations can be written as
\begin{align*}
\hat{d}_j[i]= d_{12}[i]+ w_j[i] + \eta_j[i]
\end{align*}
for $j\in \{1,2\}$.

Consequently, one will observe a reduction in the key rate. For example, with no angle observation at the eavesdropper, we know from Section~\ref{sec:beacon_power} that the key rate grows unboundedly with the beacon power $P$. However, with multipath fading, independent over different locations,
\begin{align*}
\lim_{P\to \infty} h(\hat{\dv}_{1}|\hat{\dv}_{2})&= h(\etav_1|\etav_2) \\
                                                 &\stackrel{(a)}{=} h(\etav_1)> -\infty, \vspace{-0.1in}
\end{align*}
where $\etav_1 = \{\eta_1[i]\}_{i=1}^n$, and $(a)$ follows since $\etav_1$ and $\etav_2$ are independent.
Hence, following an identical approach to Section~\ref{sec:beacon_power}, one see that $\lim_{P \to \infty} R_L < \infty$, i.e., the secret key rate remains bounded even as the power grows unboundedly.

\end{techreport}
\end{appendices}



\end{document}